%% file: p1.tex
\documentclass[runningheads]{llncs2}

\input{header}

\usepackage{caption}
\usepackage{subcaption}

\newcommand{\ceil}[1]{\lceil #1\rceil}
\newcommand{\floor}[1]{\lfloor #1\rfloor}

\begin{document}

\title{On the Problem of $p_1^{-1}$ in Locality-Sensitive Hashing}

\author{Thomas Dybdahl Ahle\orcidID{0000-0001-9747-0479}}

\authorrunning{T. Ahle}


\institute{
   IT University of Copenhagen\\
   \email{thdy@itu.dk}\\
   \url{http://thomasahle.com}
}

\maketitle

\begin{abstract}
A Locality-Sensitive Hash (LSH) function is called $(r,cr,p_1,p_2)$-sensitive, if two data-points with a distance less than $r$ collide with probability at least $p_1$ while data points with a distance greater than $cr$ collide with probability at most $p_2$.
These functions form the basis of the successful Indyk-Motwani algorithm (STOC 1998) for nearest neighbour problems.
In particular one may build a $c$-approximate nearest neighbour data structure with query time $\tilde O(n^\rho/p_1)$ where $\rho=\frac{\log1/p_1}{\log1/p_2}\in(0,1)$.
That is, \emph{sub-linear time}, as long as $p_1$ is not too small.
This is significant since most high dimensional nearest neighbour problems suffer from the curse of dimensionality, and can't be solved \emph{exact}, faster than a brute force \emph{linear-time} scan of the database.

Unfortunately, the best LSH functions tend to have very low collision probabilities, $p_1$ and $p_2$.
Including the best functions for Cosine and Jaccard Similarity.
This means that the $n^\rho/p_1$ query time of \emph{LSH is often not sub-linear after all}, even for approximate nearest neighbours!

In this paper, we improve the general Indyk-Motwani algorithm to reduce the query time of LSH to $\tilde O(n^\rho/p_1^{1-\rho})$ (and the space usage correspondingly.)
Since $n^\rho p_1^{\rho-1} < n \Leftrightarrow p_1 > n^{-1}$,
our algorithm always obtains sublinear query time, for any collision probabilities at least $1/n$.
For $p_1$ and $p_2$ small enough, our improvement over all previous methods can be \emph{up to a factor $n$} in both query time and space.

The improvement comes from a simple change to the Indyk-Motwani algorithm, which can easily be implemented in existing software packages.

\end{abstract}

\section{Introduction}
Locality Sensitive-Hashing (LSH) framework~\cite{indyk1998approximate} is one of the most efficient approaches to the nearest neighbour search problem in high dimensional spaces.
It comes with theoretical guarantees,
and it has the advantage of easy adaption to nearly any metric or similarity function one might want to search.

The $(r_1,r_2)$-near neighbour problem is defined as follows:
Given a set $X$ of points, we build a data-structure, such that given a query, $q$
we can quickly find a point $x\in X$ with distance $<r_2$ to $q$, or determine that $X$ has no points with distance $\le r_1$ to $q$.
Given a solution to this ``gap'' problem, one can obtain a $r_1/r_2$-approximate nearest neighbour data structure, or even an exact%
\footnote{In general we expect the exact problem to be impossible to solve in sub-linear time, given the hardness results of~\cite{ahle2015complexity,abboud2017distributed}.
However for practical datasets it is often possible.}
solution using known reductions~\cite{ahle2017parameter,christiani2018confirmation,datar2004locality}.

For any measure of similarity, the gap problem can be solved by LSH:
we find a distribution of functions $H$, such that $p_1\ge \Pr_{h\sim H}[h(x)=h(y)]$ when $x$ and $y$ are similar, and $p_2\le \Pr_{h\sim H}[h(x)=h(y)]$ when $x$ and $y$ are dissimilar.
Such a distribution is called $(r_1,r_2,p_1,p_2)$-sensitive.
If $p_1>p_2$ the LSH framework gives a data-structure with query time $\tilde O(n^\rho/p_1)$ for $\rho= \frac{\log1/p_1}{\log1/p_2}$, which is usually significantly faster than the alternatives.

\vspace{.5em}
\emph{At least when $p_1$ is not too small.}

\clearpage

The two most common families of LSH is Cross-Polytope (or Spherical) LSH~\cite{andoni2015practical} for Cosine similarity and MinHash~\cite{broder1998min,broder1997resemblance} for Jaccard Similarity.

Cross-Polytope is the basis of the Falconn software package~\cite{razenshteyn2018falconn},
and solves the $(r,cr)$-near neighbour problem on the sphere in time $\tilde O(n^{1/c^2}/p_1)$.
Here $p_1 = d^{-\frac{\tau^2}{4-\tau^2}}(\log d)^{-\Omega(1)}$,
where $\tau=\|p-q\|_2\in[0,2]$ is the distance between two close points.
We see that already at $\tau\approx\sqrt{2}$ (which corresponds to near orthogonal vectors) the $1/p_1$ factor results in a factor $d$ slow-down.
For larger $\tau\in(\sqrt{2},2]$ the slow-down can grow arbitrary large.
Using dimensionality reduction techniques, like the Johnson Lindenstrauss transform, one may assume $d=\eps^{-2}\log n$ at the cost of a factor $1+\eps$ distortion of the distances.
However if $\eps$ is just $1/100$, the slow-down factor of $d$ is still worse than, say, $n^{1/2}$ for datasets of size up to $10^8$, and so if $c\le \sqrt{2}$ we get that $n^\rho/p_1$ s larger than $n$.
So \emph{worse than a brute force scan of the database}!

The MinHash algorithm was introduced by Broder et al. for the Alta Vista search engine, but is used today for similarity search on sets in everything from natural language processing to gene sequencing.
MinHash solves the $(j_1,j_2)$ gap similarity search problem, where $j_1\in(0,1)$ is the Jaccard Similarity of similar sets, and $j_2$ is that of dissimilar sets, in time $\tilde O(n^\rho/j_1)$ where $\rho= \frac{\log1/j_1}{\log1/j_2}$.
(In particular MinHash is $(j_1,j_2,j_1,j_2)$-sensitive in the sense defined above.)
Now consider the case $j_1=n^{-1/4}$ and $j_2=n^{-3/10}$.
This is fairly common as illustrated in \cref{fig1}.
In this case $\rho = \frac{\log 1/j_1}{\log 1/j_2} = 5/6$,
so we end up with $n^\rho/j_1 = n^{13/12}$.
Again \emph{worse than a brute force scan of the database}!

\vspace{1em}

In this paper we reduce the query time of LSH to $n^\rho/p_1^{1-\rho}$, which is less than $n$ for all $p_1>1/n$.
In the MinHash example above, we get $n^\rho/p_1^{1-\rho} = n^{5/6 + 1/4 (1-5/6)} = n^{7/8}$.
More than a factor $n^{.208}$ improvement(!)
In general the improvement of $p_1^{-\rho}$ may be as large as a factor of $n$ when $p_1$ and $p_2$ are both close to $1/n$.
This is illustrated in \cref{fig2}.

\vspace{1em}

The improvements to LSH comes from a simple observation:
During the algorithm of Indyk and a certain ``amplification'' procedure has to be applied $\frac{\log n}{\log1/p_2}$ times.
When $\log1/p_2$ does not divide $n$, which is extremely likely, the amount of amplification has to be approximated by the next integer.
We propose instead an ensemble of LSH tables with different amounts of amplification, which when analysed sufficiently precisely yields the improvements described above.

\begin{figure}
   \centering
   \vspace*{-2cm}
   \begin{subfigure}[b]{\textwidth}
      \hspace*{-1cm}
      \centering\includegraphics[width=\textwidth]{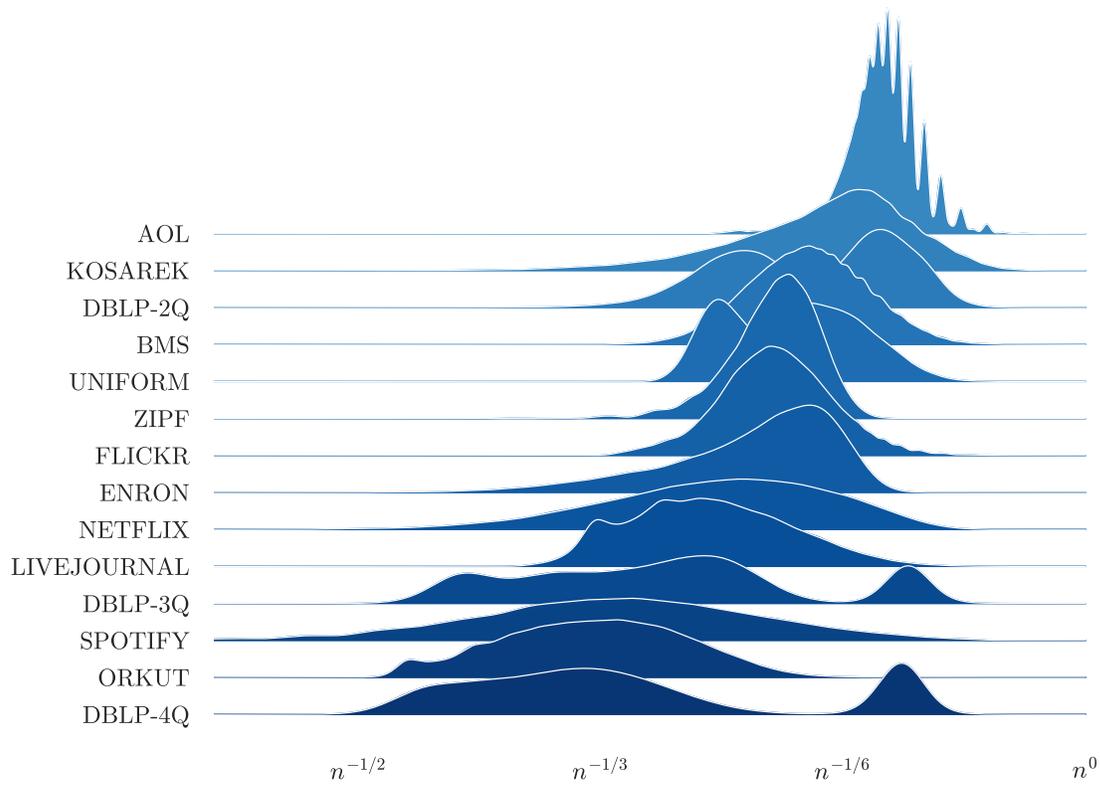}
      \caption{Density plots of pairwise Jaccard Similarities in the datasets studied by Mann et al.~\cite{mann2016empirical}.
         The similarities are normalized by the dataset sizes, so we can compare the effect of $1/p_1$ with the effect of $n^\rho$.
         We see that reasonable values for $j_1=p_1$ range between $n^{-1/3}$ and $n^{-1/6}$ on those datasets.
      }
      \label{fig1}
   \end{subfigure}
   \begin{subfigure}[b]{.8\textwidth}
      \centering\includegraphics[width=\textwidth]{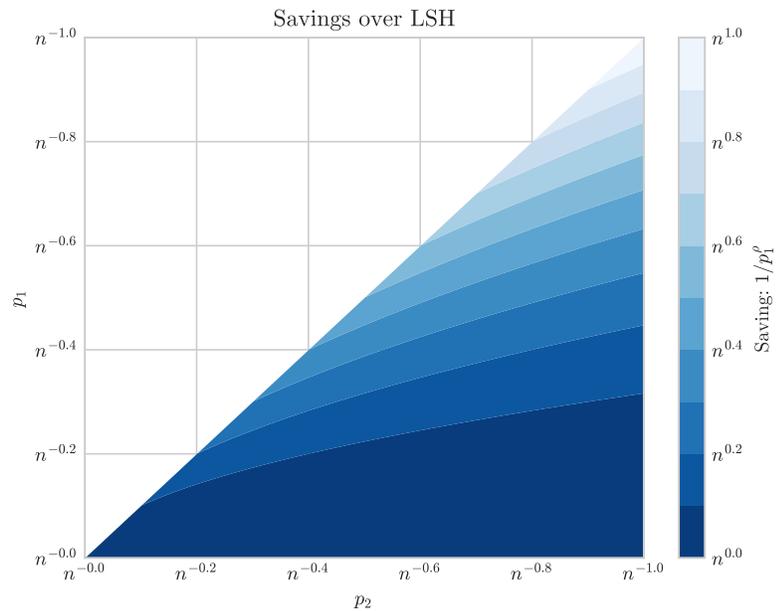}
      \caption{Saving possible, in query time and space, over classical LSH as a function of $p_1$ and $p_2$.
         With $p_1=n^{-1/4}$ and $p_2=n^{-1/3}$ we save a factor of $n^{3/16}=n^{.1875}$.
      }
      \label{fig2}
   \end{subfigure}
   \caption{Overview over available savings}
\end{figure}

\subsection{Related Work}

We will review various directions in which LSH has been improved and generalized, and how those results related to what is presented in the present article.

In many cases, the time required to sample and evaluate the hash functions dominate the time required by LSH.
Recent papers~\cite{dahlgaard2017fast,christiani2019fast} have
reduced the number of distinct calls to the hash functions which is needed.
The most recent paper in the line of work is~\cite{christiani2019fast}, which reduces the number of calls to $(\frac{\log n}{\log 1/p_2})^2/p_1$.
On top of that, however, they still require $n^{\rho}/p_1$ work, so the issue with small $p_1$ isn't touched upon.
In fact, some of the some of the algorithms in~\cite{christiani2019fast} \emph{increase} the dependency from $n^\rho/p_1$ to $n^\rho/(p_1-p_2)$.

Other work has sought to generalize the concept of Locality Sensitive Hashing to so-called Locality Sensitive Filtering, LSF~\cite{becker2016new}.
However, the best work for set similarity search based on LSF~\cite{ahle2020subsets,tobias2016} still have factors similar in spirit to $p_1^{-1}$.
E.g., the Chosen Path algorithm in~\cite{tobias2016} uses query time $\tilde O(n^\rho/b_1)$, where $b_1$ is the similarity between close sets.

A third line of work has sought to derandomize LSH.
The result is so-called Las Vegas LSH~\cite{ahle2017optimal,wei2019optimal}.
Here the families $H$ are built combinatorially, rather than at random, to guarantee the data structure always return a near neighbour, when one exists.
While these methods don't have probabilities, they still end up with similar factors for similar reasons.

As mentioned, the reason $p_1^{-1}$ shows up in all these different approaches,
is that they all rely on the same amplification procedure, which has to be applied an integer number of times.
One might wonder if tree based methods, which do an adaptive amount of amplification, could get rid of the $1/p_1$ dependency.
However as evidenced by the classical and current work~\cite{bawa2005lsh,andoni2017lsh,christiani2018confirmation,christiani2019puffinn} these methods still have a factor $1/p_1$.
We leave it open whether this might be avoidable with better analysis,
perhaps inspired by the results in this paper.

\section{Preliminaries}

Before we give the new LSH algorithm, we will recap the traditional analysis.
For a more comprehensive introduction to LSH, see the Mining of Massive Datasets book~\cite{leskovec2020mining}, Chapter 3.
In the remainder of the article we will use the notation $[n]=\{1,\dots,n\}$.

Assume we are given a $(r_1, r_2, p_1, p_2)$-sensitive LSH family, $H$, as defined in the introduction.
Let $k$ and $L$ be some integers defined later,
and let $[m]$ be the range of the hash functions, $h\in H$.
Let $n$ be an upper bound on the number of points to be inserted.
\footnote{If we don't know how many points will be inserted, several black box reductions allow transforming LSH into a dynamic data structure.}
The Indyk-Motwani data-structure consists of $L$ hash tables, each with $m^k$ hash buckets.

To insert a point, $x$, we draw $L\cdot k$ functions from $H$, denoted by $(h_{i,j})_{i\in[L],j\in[k]}$.
In each table $i\in[L]$ we insert $x$ into the bucket keyed by $(h_{i,1}(x), h_{i,2}(x), \dots, h_{i,k}(x))$.
Given a query point $q$, the algorithm iterates over the $L$ tables and retrieves the data points hashed into the same buckets as $q$.
The process stops as soon as a point is found within distance $r_1$ from $q$.

The algorithm as described has the performance characteristics listed below.
Here we assume the hash functions can be sampled and evaluated in constant time.
If this is not the case, one can use the improvements discussed in the related work.
\begin{itemize}
   \item Query time: $O(L(k+np_2^k)) = O(n^\rho p_1^{-1}\log n)$.
   \item Space: $O(nL) = O(n^{1+\rho}p_1^{-1})$ plus the space to store the data points.
   \item Success probability $99\%$.
\end{itemize}
To get these bounds, we have defined $k=\lceil\frac{\log n}{\log 1/p_2}\rceil$
and
\begin{align}
   L=\lceil p_1^{-k}\rceil
  \le \exp\left(\log1/p_1 \cdot \lceil \tfrac{\log n}{\log1/p_2}\rceil\right) + 1
  \le n^\rho/p_1+1.
\end{align}

It's clear from this analysis that the $p_1^{-1}$ factor is only necessary when $\frac{\log n}{\log 1/p_2}$ is not an integer.
However in those cases it is clearly necessary, since there is no obvious way to make a non-integer number of function evaluations.
We also cannot round $k$ down instead of up, since the number of false positives would explode:
rounding down would result in a factor of $p_2^{-1}$ instead of $p_1^{-1}$ --- much worse.

\section{LSH with High-Low Tables}

The idea of the main algorithm is to create some LSH tables with $k$ rounded down, and some with $k$ rounded up.
We call those respectively ``high probability'' tables and ``low probability'' tables.
In short ``LSH with High-Low Tables''.

The main theorem is the following:

\begin{theorem}
   Let $H$ be a $(r_1, r_2, p_1, p_2)$-sensitive LSH family, and let $\rho = \frac{\log 1/p_1}{ \log 1/p_2}$.
   Assume $p_1>1/n$ and $p_2>1/n$.
   Then there exists a solution to the $(r_1, r_2)$-near neighbour problem with the following properties:
   \begin{itemize}
      \item Query time: $O(n^\rho p_1^{\rho-1}\log n)$.
      \item Space: $O(nL) = O(n^{1+\rho}p_1^{\rho-1})$ plus the space to store the data points.
      \item Success probability $99\%$.
   \end{itemize}
\end{theorem}

\begin{proof}
Assume $r_1, r_2, p_1, p_2$ are given.
Define $\rho = \frac{\log1/p_1}{\log1/p_2}$,
$\kappa=\frac{\log n}{\log1/p_2}$,
and $\alpha = \ceil\kappa - \kappa \in [0,1)$.
We build $\floor a + \ceil b$ tables (for $a,b\ge 0$ to be defined), where the first $\floor a$ use the hash function concatenated $\floor \kappa$ times as keys, and the remaining $\ceil b$ use it concatenated $\ceil \kappa$ times.

The total number of tables to query is then $\floor a + \ceil b$.
The expected total number of far points we have to retrieve is
\begin{align}
   n(\floor a p_2^{\floor\kappa} + \ceil b p_2^{\ceil\kappa})
   &= n(\floor a p_2^{\kappa-1+\alpha}+\ceil b p_2^{\kappa+\alpha})
 \\&= \floor a p_2^{-1+\alpha}+\ceil b p_2^{\alpha}
 \\&\le a p_2^{-1+\alpha}+ (b+1) p_2^{\alpha}
 \\&\le a p_2^{-1+\alpha}+ b p_2^{\alpha} + 1
 \label{eq:expected-collisions}
 .
\end{align}
For the second equality, we used the definition of $\kappa$: $p_2^\kappa=1/n$.
We only count the expected number of points seen that are at least $r_2$ away from the query.
This is because the algorithm, like classical LSH, terminates as soon as it sees a point with distance less than $r_2$.

Given any point in the database within distance $r_1$ we must be able to find it with high enough probability.
This requires that the query and the point shares a hash-bucket in one of the tables.
The probability that this is not the case is
\begin{align}
   (1-p_1^{\floor \kappa})^{\floor a}
   (1-p_1^{\ceil \kappa})^{\ceil b}
   &\le
   (1-p_1^{\floor \kappa})^{a-1}
   (1-p_1^{\ceil \kappa})^{b}
   \\&\le
   \exp(-a p_1^{\floor \kappa} - b p_1^{\ceil \kappa}) (1-p_1^{\floor \kappa})^{-1}
 \\&=
 \exp(- (a p_1^{-1+\alpha} + b p_1^{\alpha})n^{-\rho})
(1-p_1^{\floor \kappa})^{-1}
 \\&\le
 \exp(- (a p_1^{-1+\alpha} + b p_1^{\alpha})n^{-\rho})\cdot 2
.
\end{align}
For the equality, we used the definition of $\kappa$ and $\rho$: $p_1^\kappa = p_2^{\rho\kappa} = n^{-\rho}$.
For the last inequality we have assumed $p_2>1/n$ so $\floor \kappa\ge 1$,
and that $p_1<1/2$, since otherwise we could just get the theorem from the classical LSH algorithm.

\vspace{1em}

We now define $a, b$ such that 
$a p_2^{-1+\alpha}+ b p_2^{\alpha} =a+b$
and $a p_1^{-1+\alpha} + b p_1^{\alpha}=n^\rho$.
By the previous calculations this will guarantee the number of false positives is not more than the number of tables, and a constant success probability.

We can achieve this by taking
\begin{align}
   \begin{bmatrix}
      a\\b
   \end{bmatrix}
   =
   \begin{bmatrix}
      p_2^{-1+\alpha}-1 & p_2^{\alpha}-1
      \\
      p_1^{-1+\alpha} & p_1^{\alpha}
   \end{bmatrix}^{-1}
   \begin{bmatrix}
      0 \\ n^\rho
   \end{bmatrix}
   =
   \frac{n^\rho
      }{
      (p_2^{-1+\alpha}-1) p_1^{\alpha}
      +
      (1-p_2^{\alpha}) p_1^{-1+\alpha}
   }
      \begin{bmatrix}
         1-p_2^\alpha
         \\
         p_2^{-1+\alpha}-1
      \end{bmatrix}
      .
\end{align}
We note that both values are non-negative, since $\alpha\in[0,1]$.

When actually implementing the LSH algorithm width High-Low buckets, these are the values you should use for the number of respectively the high and low probability tables.
That will ensure you take full advantage of when $\alpha$ is not worst case, and you may do even better than the theorem assumes.

To complete the theorem we need to prove $a+b\le n^\rho p_1^{\rho-1}$.
For this we bound
\begin{align}
   \frac{a+b}{n^\rho}
   &= \frac{p_2^{-1+\alpha}-p_2^\alpha}{ (p_2^{-1+\alpha}-1) p_1^{\alpha} + (1-p_2^{\alpha}) p_1^{-1+\alpha} }
 \\&\le
   \left(\frac{(p_1-p_2)\log1/p_1}{(1-p_1)\log p_1/p_2}\right)^\rho
   \left(\frac{(1-p_2)\log p_1/p_2}{(p_1-p_2)\log1/p_2}\right)
   \label{eq:a-bound}
 \\&=
 \exp\left(D\left(\rho \,\bigg\|\, \frac{1/p_1-1}{1/p_2-1}\right)\right)
 \\&\le
 p_1^{\rho-1}
 .
\end{align}
Here $D(r \| x) = r \log\tfrac{r}{x} + (1-r)\log\tfrac{1-r}{1-x}$ is the Kullback-Leibler divergence.
The two inequalities are proven in the Appendix as \cref{lem:alpha} and \cref{lem:divrx}.
The first bound comes from maximizing over $\alpha$,
so in principle we might be able to do better if $\kappa=\frac{\log n}{\log 1/p_2}$ is close to an integer.
The second bound is harder, but the realization that the left hand side can be written on the form of a divergence helps a lot.
The bound is tight up to a factor 2, so no significant improvement is possible.

Finally we can boost the success probability from $1-2/e\approx 0.26$ to 99\% by repeating the entire data-structure 16 times.
\end{proof}

\bibliographystyle{splncs04}
\bibliography{p1}

\section{Appendix}
\appendix

\begin{lemma}\label{lem:alpha}
   For all $\alpha\in[0,1]$ we have
   \begin{align}
      f(\alpha)
   = \frac{p_2^{-1+\alpha}-p_2^\alpha}{ (p_2^{-1+\alpha}-1) p_1^{\alpha} + (1-p_2^{\alpha}) p_1^{-1+\alpha} }
   \le
   \left(\frac{(p_1-p_2)\log1/p_1}{(1-p_1)\log p_1/p_2}\right)^\rho
   \left(\frac{(1-p_2)\log p_1/p_2}{(p_1-p_2)\log1/p_2}\right)
   ,
   \end{align}
   where $\rho=\frac{\log1/p_1}{\log1/p_2}$.
\end{lemma}
\begin{proof}
   We first show that $f(\alpha)$ is log-concave, which implies it is maximized at the unique $\alpha$ such that $f'(\alpha)=0$.
   Log-concavity follows easily by noting
   \begin{align}
      \frac{d^2  \log f(\alpha)}{d\alpha^2}
      = 
      -\frac{(1-p_1)(p_1-p_2)p_2^{1+\alpha}(\log\tfrac1{p_2})^2}
      {((1-p_1)p_2+p_2^\alpha(p_1-p_2))^2}
      \le 0
      .
   \end{align}

   Meanwhile
   \begin{align}
      \frac{d f(\alpha)}{d\alpha}
      = 
      \frac{p_1(1-p_2)(p_2/p_1)^\alpha}
      {((1-p_1)p_2+p_2^\alpha(p_1-p_2))^2}
      \left[(p_1-p_2)p_2^\alpha \log\tfrac1{p_1}-(1-p_1)p_2\log\tfrac{p_2}{p_1}\right]
      ,
   \end{align}
   which implies $f(\alpha)$ is maximized in
   \begin{align}
      \alpha = \log\frac{(1-p_1)p_2\log\tfrac{p_2}{p_1}}{(p_1-p_2) \log\tfrac1{p_1}} \big/ \log p_2
      .
   \end{align}
   Plugging into $f$ yields the lemma.
\end{proof}
Note that $f(\alpha)$ is not regularly concave as $p_1$ and $p_2$ gets small enough.
Hence the use of log-concavity is necessary.

Next, we state a useful inequality, which is needed for the last proof.
\begin{lemma}\label{lem:p1r}
   Let $p,r\in(0,1)$, then
   \begin{align}
      1-\frac{1-p}{r} \le
      p^{1/r}\le \frac{p r}{1-p(1-r)}
      .
   \end{align}
\end{lemma}
\begin{proof}
   We have $\frac{d^2}{d p^2} p^{1/r} = p^{1/r} (p r)^{-2} (1-r)$, so $p^{1/r}$ is convex as a function of $p$.
   Since $1-\frac{1-p}{r}$ is it's tangent (at $p=1$) we get the first inequality.

   For the second inequality,
   define $g(p) = p^{1/r} / \frac{p r}{1-p(1-r)}$.
   Then $g(1) = 1$ and $g(p)$ is non-decreasing, since
   \begin{align}
      g'(p)
      = p^{1/r} (p r)^{-2} (1-p)(1-r)
      \ge 0
      .
   \end{align}
   This shows that for $p\le 1$ we have $p^{1/r} / \frac{p r}{1-p(1-r)} \le 1$,
   which is what we wanted to prove.
\end{proof}

\begin{lemma}\label{lem:divrx}
   Let $p,r\in(0,1]$
   and let $x= \frac{1/p\,-\,1}{1/p^{1/r}-1}$,
   then
   \begin{align}
      D(r \| x)
      \le
      r\log\tfrac{r}{x}
      \le (1-r)\log\tfrac{1}{p}
      ,
      \label{eq:Dbound}
   \end{align}
   where $D(r\|x)=r \log\tfrac{r}{x}+(1-r)\log\tfrac{1-r}{1-x}$.
\end{lemma}
\begin{proof}
   Using the upper bound of \cref{lem:p1r} it follows directly that $x\in(0,r)$.
   This suffices to show the first inequality of \eqref{eq:Dbound}, since for $x\le r$ we have $\frac{1-r}{1-x}\le 1$ and so the second term of $D(r \| x)$ is non-positive.

   \vspace{.5em}

   For the second inequality,
   we note that it is equivalent after manipulations to $x \ge r p^{1/r-1}$.
   Plugging in $x$, and after more simple manipulations,
   that's equivalent in the range to $p^{1/r}\ge1-\frac{1-p}{r}$,
   which is \cref{lem:p1r}.

   This finishes the proof of \cref{lem:divrx}.
\end{proof}
It's somewhat surprising that the last argument in the proof of \cref{lem:divrx} works,
since if we had plugged the lower bound from \cref{lem:p1r} directly into the problem we would have had
\begin{align}
   r\log\tfrac{r}{x}
   \le r\log\tfrac{p r}{p+r-1}
   ,
\end{align}
which is much weaker than what we prove, and not even defined for $p+r<1$.
\end{document}

%% file: header.tex
\usepackage[margin=1.3in]{geometry}
\usepackage[utf8]{inputenc}
\usepackage{microtype}

\usepackage{amsmath}
\usepackage{amsfonts}
\usepackage{amssymb}

\usepackage{todonotes}
\usepackage{enumitem}
\usepackage{multicol}
\usepackage{hyperref}
\usepackage{cleveref}
\usepackage{autonum}

\allowdisplaybreaks

\newcommand{\eps}{\varepsilon}